\documentclass[conference,letterpaper]{IEEEtran}
\usepackage{amsmath,amssymb,amsthm,cite}
\usepackage{graphicx}
\usepackage{color}
\usepackage{subcaption}
\usepackage{pbox}
\usepackage{tkz-euclide}
\usetkzobj{all}
\usepackage{tikz}
\usepackage{url}
\usepackage{todonotes}
\usepackage[bookmarks=false,colorlinks=true,urlcolor=blue,citecolor=blue,linkcolor=blue]{hyperref}
\usepackage{cleveref}
\usepackage{comment}
\usepackage{flushend}

\usepackage[left=0.75in,right=0.75in, top=0.75in,bottom=0.75in]{geometry}

\newtheorem{lem}{Lemma}
\newtheorem{thm}{Theorem}
\newtheorem{defn}{Definition}

\newtheorem{rem}{Remark}

\graphicspath{{Figures/}}

\crefname{equation}{}{}
\Crefname{equation}{}{}
\crefname{thm}{theorem}{theorems}
\Crefname{thm}{Theorem}{Theorems}
\crefname{clm}{claim}{claims}
\Crefname{clm}{Claim}{Claims}
\Crefname{coro}{Corollary}{Corollaries}
\Crefname{lem}{Lemma}{Lemmas}
\Crefname{sec}{Section}{Sections}
\crefname{app}{appendix}{appendices}
\Crefname{app}{Appendix}{Appendices}
\crefname{prop}{proposition}{propositions}
\Crefname{prop}{Proposition}{Propositions}
\Crefname{propty}{Property}{Properties}
\crefname{figure}{fig.}{figures}
\Crefname{figure}{Fig.}{Figures}
\crefname{defn}{definition}{definitions}
\Crefname{defn}{Definition}{Definitions}
\crefname{fact}{fact}{facts}
\Crefname{fact}{Fact}{Facts}
\crefname{appendix}{appendix}{appendices}
\Crefname{appendix}{Appendix}{Appendices}
\crefname{algo}{algorithm}{algorithms}
\Crefname{algo}{Algorithm}{Algorithms}
\crefname{algorithm}{algorithm}{algorithms}
\Crefname{algorithm}{Algorithm}{Algorithms}
\crefname{conj}{conjecture}{conjectures}
\Crefname{conj}{Conjecture}{Conjectures}
\crefname{obs}{observation}{observations}
\Crefname{obs}{Observation}{Observations}

\usepackage[inline]{trackchanges}
\addeditor{mehmet}
\addeditor{sarah}
\addeditor{ann}
\addeditor{gauri}
\addeditor{swanand}
\addeditor{gretchen}
\addeditor{carolyn}
\addeditor{emina}

\begin{document}

\newcommand{\totalCaches}{m}
\newcommand{\nFork}{n}
\newcommand{\nPartialFork}{r}
\newcommand{\kJoin}{k}
\newcommand{\xm}{x_m}

\newcommand{\Si}[1]{S_{#1}} 
\newcommand{\Sij}[2]{S_{#1}^{#2}} 
\newcommand{\lami}[1]{\lambda_{#1}} 
\newcommand{\lamij}[2]{\lambda^{(#1)}_{#2}} 
\newcommand{\rmv}[1]{}

\newcommand{\Rij}[2]{R^{(#1)}_{#2}}

\newcommand{\tMDS}{t_{\textrm{MDS}}} 

\newtheorem{claim}{Claim}
\newtheorem{conjecture}{Conjecture}
\newtheorem{lemma}{Lemma}
\newtheorem{theorem}{Theorem}
\newtheorem{corollary}{Corollary}


\title{On the Service Capacity Region of Accessing Erasure Coded Content}


\author{
    \IEEEauthorblockN{Mehmet~Akta\c{s}\IEEEauthorrefmark{1}, Sarah~E.~Anderson\IEEEauthorrefmark{2}, Ann~Johnston\IEEEauthorrefmark{3}, Gauri~Joshi\IEEEauthorrefmark{4}, \\ Swanand~Kadhe\IEEEauthorrefmark{5}, Gretchen L.~Matthews\IEEEauthorrefmark{6}, Carolyn~Mayer\IEEEauthorrefmark{7}, and Emina~Soljanin\IEEEauthorrefmark{1}} \\
    \IEEEauthorblockA{\IEEEauthorrefmark{1}Rutgers University \{mehmet.aktas, emina.soljanin\}@rutgers.edu, \IEEEauthorrefmark{2}University of St. Thomas \{ande1298\}@stthomas.edu}
    \IEEEauthorblockA{\IEEEauthorrefmark{3}Penn State University \{abj5162\}@psu.edu, \IEEEauthorrefmark{4}Carnegie Mellon University \{gaurij\}@andrew.cmu.edu}
    \IEEEauthorblockA{\IEEEauthorrefmark{5}Texas A\&M University \{swanand.kadhe\}@tamu.edu, \IEEEauthorrefmark{6}Clemson University \{gmatthe\}@clemson.edu}
    \IEEEauthorblockA{\IEEEauthorrefmark{7}University of Nebraska - Lincoln \{cmayer\}@huskers.unl.edu}
}

\maketitle


\begin{abstract}
Cloud storage systems generally add redundancy in storing content files such that $K$ files are replicated or erasure coded and stored on $N > K$ nodes. In addition to providing reliability against failures, the redundant copies can be used to serve a larger volume of content access requests. A request for one of the files can either be sent to a systematic node, or one of the repair groups. In this paper, we seek to maximize the service capacity region, that is, the set of request arrival rates for the $K$ files that can be supported by a coded storage system. We explore two aspects of this problem: 1) for a given erasure code, how to optimally split incoming requests between systematic nodes and repair groups, and 2) choosing an underlying erasure code that maximizes the achievable service capacity region. In particular, we consider MDS and Simplex codes. Our analysis demonstrates that erasure coding makes the system more robust to skews in file popularity than simply replicating a file at multiple servers, and that coding and replication together can make the capacity region larger than either alone.
\end{abstract}


\section{Introduction}
\label{sec:intro}







%
%


Cloud storage systems are expected to provide reliability against failures and ensure availability of stored content during high demand, while handling massive amount of data. 
In order to combat failures, redundancy is added using either replication or erasure coding. Even though replication has conventionally been preferred due its simplicity, a large body of recent literature has proposed novel erasure coding techniques as a more efficient way to provide reliability, see e.g.,~\cite{Dimakis:10,Dimakis-Survey:11,Huang:07,gopalan2012locality}.
%
In addition to reliability, redundancy has been shown to be effective in enhancing availability by reducing download latency for retrieving entire data in a number of recent research papers, see e.g.,~\cite{joshi2012coding, shah2014mds, liang2014fast, gardner2015reducing}. On the other hand, for downloading hot data, wherein users are interested in downloading individual files with different popularities, the role of erasure codes in reducing latency is not yet well understood, and is a topic of active research, 
see~\cite{swanand_isit_2015,swanand_allerton_2015,Simplex:AktasNS17}. 

Besides download latency, an important metric that measures the availability of the stored data is the service capacity region, which is the space of download request rates for which the system is stable. In comparison to download latency, this metric of service capacity has received very little attention. One notable exception is the work of~\cite{Noori:16}, in which the authors study storage allocation strategies to maximize the service rate for downloading entire data.





In this paper, we seek to investigate the effect of redundancy on the service capacity region of the the system. To the best of our knowledge, this is the first work to investigate the service capacity for downloading hot data from coded storage.
More specifically, we consider a system with $K$ files, $f_1, f_2, \dots , f_K$ that are replicated or stored in coded form on $N > K$ servers. Requests to download file $f_i$ arrive at rate $\lambda_i$. Our first objective is to maximize the set of arrival rates supported by a given coding scheme. Next, we compare service capacity regions for different coding schemes.

Together with replication and maximum distance separable (MDS) codes, we consider an important family of distributed storage codes called {\it availability codes} (see~\cite{Wang:14,Rawat:14Availability,Tamo:14Availability}). Availability codes enable any codeword symbol to be recovered from multiple, disjoint subsets of other symbols of small size. Amongst availability codes, we focus our attention on the special sub-class, namely {\it simplex codes} due to their optimality in rate 
\cite{CadambeM:15}.

We note that MDS codes are more robust to handling variations in the access patterns compared to replication. Availability codes handle the skews in popularities better than MDS codes. Surprisingly, {\it hybrid codes}, formed by replicating some of the files and adding MDS parity symbols perform exceptionally well by achieving a large service capacity region. 


It is important to note that, even though we focus on the service capacity for content download, the techniques are also applicable for analyzing service capacity for coded computation. For example, suppose some users are interested in computing a matrix vector product $AX$, while others are interested in computing $BX$. Suppose two worker nodes store matrices $A$ and $B$ respectively, while the third worker node stored the sum $A+B$ of the two matrices, (assuming that $A$ and $B$ are of the same size). Then, excess number of requests to compute $AX$ can be satisfied by using the other two workers.

{\it Organization:} 
In Section~\ref{sec:prob_formu}, we describe the problem setup and formally define the service capacity region. In Section~\ref{sec:examples}, we motivate the analysis of service capacity by computing the service capacity for several small examples of codes. In Section~\ref{sec:n_k_mds}, we focus on systematic $(N,K)$ MDS codes. We find an outer bound on the service capacity region, and present a greedy algorithm referred to as waterfilling algorithm. We show that the waterfilling algorithm is optimal, and it achieves the outer bound for MDS codes of rate smaller than or equal to half. In Section~\ref{sec:simplex}, we characterize the service capacity of simplex codes. Our proof for converse uses an interesting connection to graph covering. In Section~\ref{sec:adding_sys_nodes} we consider {\it hybrid codes} consisting of replication and MDS parities. For $K=2$ files, we characterize the service capacity of hybrid codes as a function of the number of replicas of the two files and the number of MDS parities.

\section{Problem Formulation}
\label{sec:prob_formu}
We have $K$ files, $f_1, f_2, \dots f_K$ of equal size stored redundantly across $N$ nodes, labeled $1$ through $N$. We refer to the coding scheme encoding $K$ files into $N$ as an $(N,K)$ code. Requests to download $f_i$ arrive at rate $\lambda_i$. Our objective is to determine the set of arrival rates $(\lambda_1, \dots \lambda_K)$ that can be served by the system. We refer to space of arrival rates that can be served as the {\it capacity region} of the system. 

As the coding scheme adds redundancy, each file can be recovered in multiple ways. For a file $f_i$, a subset of nodes (of minimal size) from which the file can be recovered is referred to as a {\it recovering set} of $f_i$. We denote the number of distinct recovering sets of $f_i$ as $t_i$, and label them as $\Rij{i}{1}, \cdots, \Rij{i}{t_i}$. For example, consider the following $(4,2)$ code over $\mathbb{F}_{3}$: $\{f_1, f_2, f_1+f_2, f_1+2f_2\}$. There are four recovering sets for each file. Recovering sets of $f_1$ are given as $\Rij{1}{1} = \{1\}$, $\Rij{1}{2} = \{2,3\}$, $\Rij{1}{3} = \{2,4\}$, and $\Rij{1}{4} = \{3,4\}$. Observe that for a systematic $(N,K)$ MDS code, there are $\binom{N-1}{K}+1$ recovering sets for every file. 

We consider the class of scheduling strategies that assign a fraction of requests for a file to each of its recovering sets. Let $\lamij{i}{j}$ be the fraction of requests for file $f_i$ that are assigned to its recovering set $\Rij{i}{j}$. Note that $\sum_{j=1}^{t_i}\lamij{i}{j} = \lami{i}$. Then, the service capacity region of an $(N,K)$ coding scheme is defined as follows.

\begin{defn}[Service Capacity Region]
\label{def:service-capacity-region}
Consider a system storing $K$ files over $N$ nodes using an $(N,K)$ code such that a file $f_i$ has $t_i$ recovering sets $\Rij{i}{1},\cdots,\Rij{i}{t_i}$. Let the service rate of every node is $\mu$. Then, the service capacity region of such a system is the set of vectors $(\lami{1}, \ldots, \lami{K})$ such that, for every $1\leq i\leq K$, there exist $\lamij{i}{j}$, $1\leq j\leq t_i$, satisfying the following:
  \begin{IEEEeqnarray}{rCl}
  \label{eq:LP-sum-constraint}
  \sum_{j=0}^{\tMDS} \lamij{i}{j} & = & \lami{i}, \quad 1\leq i \leq K\\
  \label{eq:LP-UB-constraint}
  \sum_{i=1}^{K} \sum_{j:\ell \in \Rij{i}{j}} \lamij{i}{j} & \leq & \mu, \quad 1\leq \ell\leq N\\
  \label{eq:LP-non-negativity-constraint}
  \lamij{i}{j} & \geq & 0, \quad  1\leq i\leq K,\:  1\leq j\leq t_i.
  \end{IEEEeqnarray}
\end{defn}

Note that, given any $K-1$ arrival rates $\lami{i_1},\cdots,\lami{i_{K-1}}$, finding the maximum value of $\lami{i_K}$ and the allocations $\lamij{i_l}{j}$ such that~\eqref{eq:LP-sum-constraint},~\eqref{eq:LP-UB-constraint}, and~\eqref{eq:LP-non-negativity-constraint} hold can be considered as a constrained optimization problem. Specifically, given $\lami{1},\cdots,\lami{K-1}$, the linear program to compute the maximum $\lami{K}$ is described as follows.   

\begin{IEEEeqnarray}{rC}
\max & \lami{K} = \sum_{j=0}^{\tMDS}\lamij{K}{j}\nonumber\\
\textrm{s.t.} & ~\eqref{eq:LP-sum-constraint},~\eqref{eq:LP-UB-constraint},~\eqref{eq:LP-non-negativity-constraint}.\nonumber
\end{IEEEeqnarray}

\section{Examples of Service Capacity Regions}
\label{sec:examples}

To motivate the analysis, suppose $K=2$, and we have two files $a$ and $b$ which are stored on $N=4$ nodes. We compare three storage schemes: uncoded, MDS coded, and a hybrid between repetition and coding shown in \Cref{fig:4_2_cache_eg}. 

\begin{figure}[t]
\begin{subfigure}[t]{0.85\linewidth}
     \centering
    \includegraphics[width= 2.0 in]{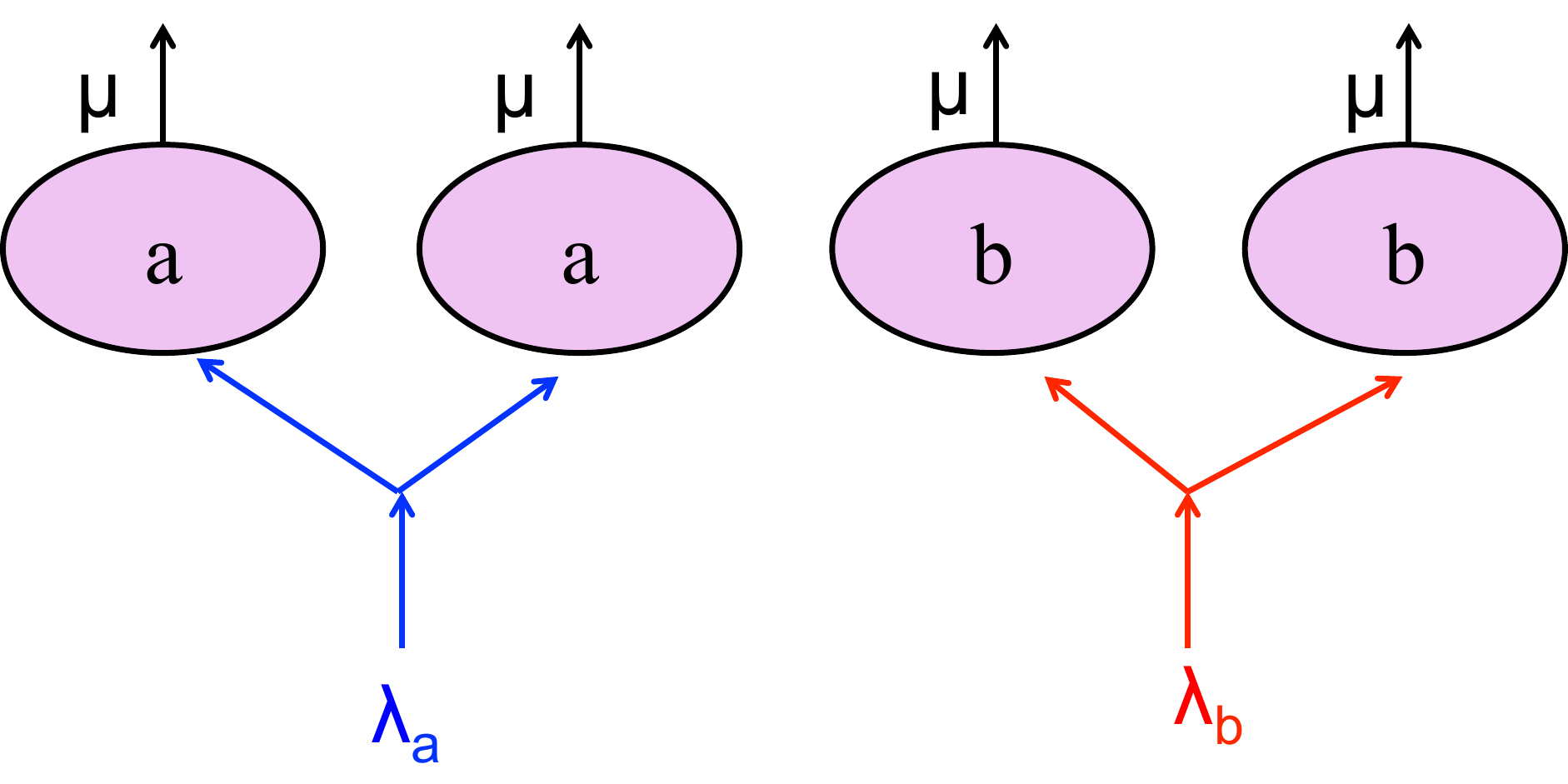}
 \caption{$(4,2)$ repetition coded system\label{fig:4_2_rep_eg}}
 \end{subfigure}
 \begin{subfigure}[t]{0.85\linewidth}
    \centering
   \includegraphics[width= 2.0 in]{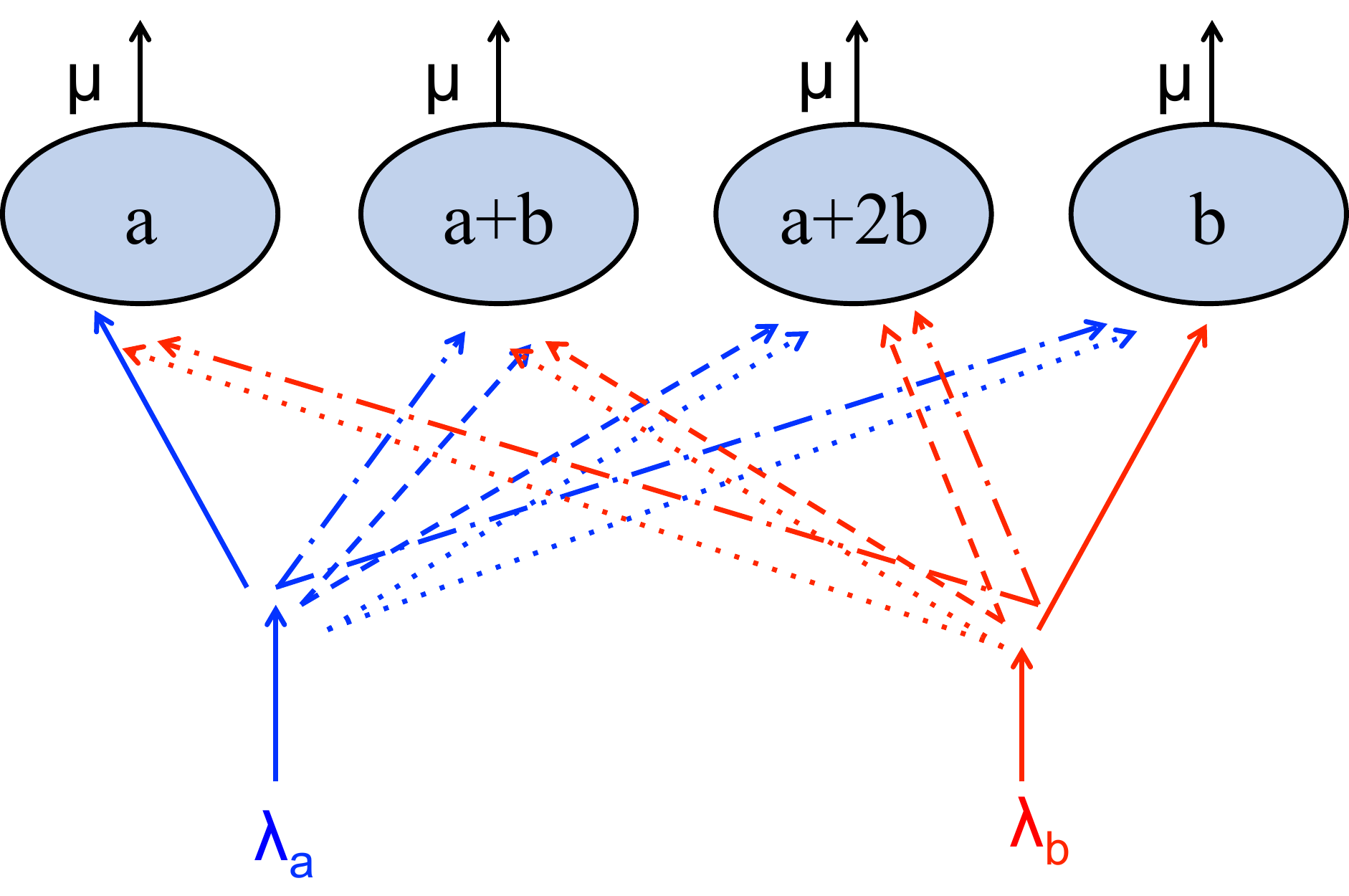}
\caption{$(4,2)$ erasure coded system.\label{fig:4_2_coded_eg}}
\end{subfigure}
  \begin{subfigure}[t]{0.85\linewidth}
     \centering
    \includegraphics[width= 2.0 in]{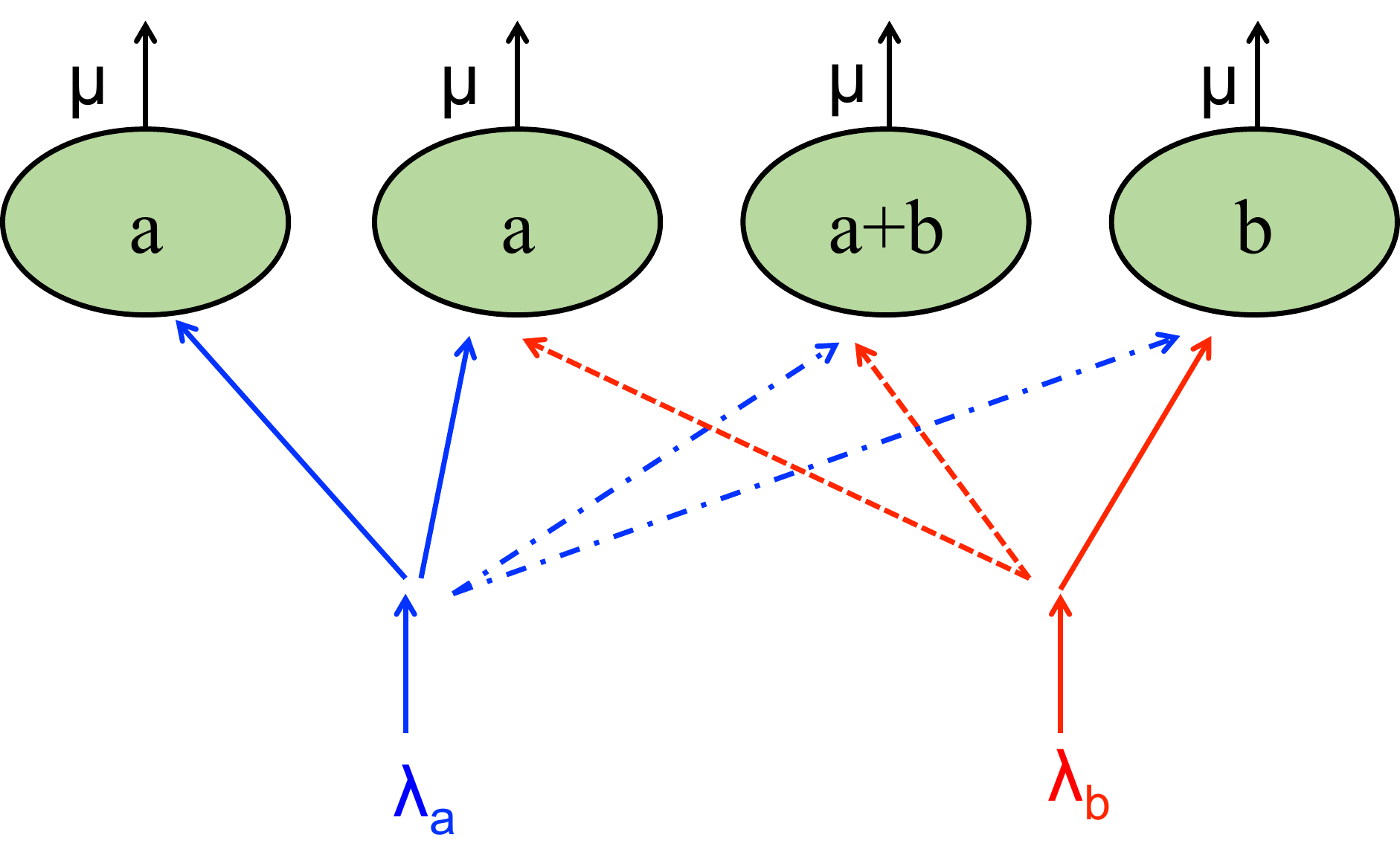}
 \caption{$(4,2)$ hybrid system, favoring $a$\label{fig:4_2_hybrid_eg}}
 \end{subfigure}
\caption{Different ways of hosting $k=2$ files on $n=4$ nodes. \label{fig:4_2_cache_eg}}
\vspace{-0.5cm}
\end{figure}

\subsection{Repetition Coding}
\label{subsec:repetition_eg}
Consider the uncoded system where file $a$ and $b$ are replicated at $2$ servers each, illustrated in \Cref{fig:4_2_rep_eg}. Since each server can support rate $\mu$ of arrivals, we have constraints $\lambda_a \leq 2\mu$ and $\lambda_b \leq 2 \mu$. Thus the achievable rate region is the square $0 \leq \lambda_a, \lambda_b \leq 2\mu$ illustrated in pink in \Cref{fig:4_2_rate_region}.

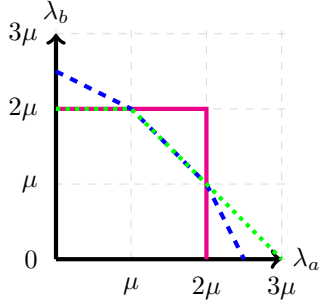
\begin{figure}[t]
\begin{center}
\begin{tikzpicture}[ultra thick,scale=1]
\draw[help lines, color=gray!30, dashed] (-.01,-.01) grid (3.1,3.1);
\draw[->,ultra thick,black] (0,0)--(3,0) node[right]{$\lambda_a$};
\draw[->,ultra thick,black] (0,0)--(0,3) node[above]{$\lambda_b$};
\draw[ultra thick,magenta] (0,2)--(2,2)--(2,0);
\draw[ultra thick,blue,dashed] (0,2.5)--(1,2)--(2,1)--(2.5,0);
\draw[ultra thick,green,dotted] (0,2)--(1,2)--(3,0);
\node[circle,scale=.5,label=left:$0$] (O) at (0,0) {};
\node[circle,scale=.5,label=below:$\mu$] (x1) at (1,0) {};
\node[circle,scale=.5,label=below:$2\mu$] (x2) at (2,0) {};
\node[circle,scale=.5,label=below:$3\mu$] (x3) at (3,0) {};
\node[circle,scale=.5,label=left:$\mu$] (y1) at (0,1) {};
\node[circle,scale=.5,label=left:$2\mu$] (y2) at (0,2) {};
\node[circle,scale=.5,label=left:$3\mu$] (y3) at (0,3) {};
\label{fig:4_2_rate_region}
\end{tikzpicture}
\end{center}
\caption{Service capacity regions of the $(4,2)$ repetition (in pink), MDS coded (in blue) and hybrid (in green) systems.}
\end{figure}

\subsection{MDS Coding} 
\label{subsec:mds_eg}
Next let us find the rate region of the $(4,2)$ coded system illustrated in \Cref{fig:4_2_cache_eg}. Recall for each given $\lambda_a$, we want to determine the maximum achievable $\lambda_b$. We divide the problem into three cases: \\
\textit{Case 1 ($0 \leq \lambda_a < \mu$):} All the requests for file $a$ should be assigned to the systematic node $a$. Requests for file $b$ can utilize the remaining capacity $\mu -\lambda_a$ of this node. We can use this by assigning $(\mu -\lambda_a)/2$ requests for file $b$ to nodes $a$ and $a+b$, and $(\mu - \lambda_a)/2$ to nodes $a$ and $a+2b$. Now nodes $a+b$ and $a+2b$ have $\mu/2 + \lambda_a/2$ capacity each remaining, which can be used to serve $\mu/2 + \lambda_a/2$ requests per second. Thus, maximum achievable $\lambda_b$ is
\begin{align}
 \lambda_b &= \mu + (\mu - \lambda_a) + (\mu + \lambda_a)/2 \\
 &= 2.5 \mu - \lambda_a/2
 \end{align}
\textit{Case 2 ($\mu \leq \lambda_a < 2 \mu $):} \\
Out of $\lambda_a$, $\mu$ volume of requests are assigned to the systematic node $a$. The remaining $\lambda_a - \mu$ traffic is assigned to nodes $a+b$ and $a + 2b$ from which we can recover file $a$. Thus, the coded nodes $a+b$ and $a+2b$ have $\mu - (\lambda_a -\mu)$ capacity remaining to serve requests for file $b$. Hence, the maximum achievable $\lambda_b$ is
\begin{align}
\lambda_b &= \mu + \mu - (\lambda_a - \mu),\\
&= 3 \mu - \lambda_a.
\end{align}

\textit{Case 3 ($\lambda_a \geq 2 \mu$):}
The solution to this case is same as Case 1, with $\lambda_a$ replaced by $\lambda_b$. Thus the maximum achievable $\lambda_b$ is
\begin{align}
\lambda_b &= 5 \mu - 2 \lambda_a
\end{align}

Combining these cases, we get the achievable rate region illustrated in blue in Fig. 2. 


\subsection{Hybrid Coding} 
\label{subsec:hybrid_eg}
If file $a$ is known to be more popular than $b$, we can have a coded system with $N=4$ nodes storing $a$, $a$, $a+b$ and $b$ respectively. This coding scheme is a combination of repetition and erasure coding. We can find the service capacity by dividing the problem into cases, similar to \Cref{subsec:mds_eg}. For this system the service capacity region is given by Fig. 2. 

\section{$(N,K)$ systematic MDS coded systems}
\label{sec:n_k_mds}
In this section we find the service capacity region of a system of $N$ servers that store $K$ files $f_1, \ldots, f_K$ together with $N-K$ parity files that are generated using an $(N,K)$ MDS code. Each of the original and redundant files are distributed across all $N$ servers. Each file $f_i$ can be downloaded from the server storing it, which we refer to as the systematic server for the file, or by accessing any $K$ of the remaining $N-1$ servers.

Let the arrival rate of requests for file $f_i$ be denoted by $\lambda_i$. We want to determine the set of arrival rate vectors $(\lambda_1, \dots \lambda_K)$ that can supported by the system. 

\subsection{Outer Bound on the Rate Region}
First we find an outer bound of the service capacity region.
\begin{thm}[Outer Bound]
  The set of all achievable request vectors $(\lambda_1, \lambda_2, \dots, \lambda_K)$ lies inside the region described by
  \begin{align}
    \sum_{i=1}^{K} \left( \min(\lambda_i, \mu) +  K(\lambda_i - \mu)^{+} \right) \leq N\mu, \label{eqn:mds_outer_bnd}
  \end{align}
  where the notation $(x)^{+} = \max(0, x)$.
\label{thm:mds_outer_bnd}
\end{thm}
\begin{proof}
  Each server in the system can support $\mu$ volume of requests, and thus the total capacity is $N\mu$. We now determine the total system capacity utilized by file download requests, and ensure that it is less than $N\mu$. Downloading a file from $K$ coded servers requires downloading data of size $K$ times the file size. If $\lambda_i$ is the rate of request arrivals for file $f_i$, the minimum system capacity utilized by these requests is  $\min(\lambda_i, \mu) +  K(\lambda_i - \mu)^{+}$. Since the total system capacity is $N\mu$, the sum of the capacity utilized by all requests must be less than $N\mu$. Thus we have \eqref{eqn:mds_outer_bnd}.
\end{proof}

\begin{rem}
For the $(4,2)$ system, the region described by \Cref{thm:mds_outer_bnd} matches exactly with the achievable region found in \Cref{fig:4_2_rate_region}.
\end{rem}


\begin{figure}[t]
    \centering
   \includegraphics[width= 3.1in]{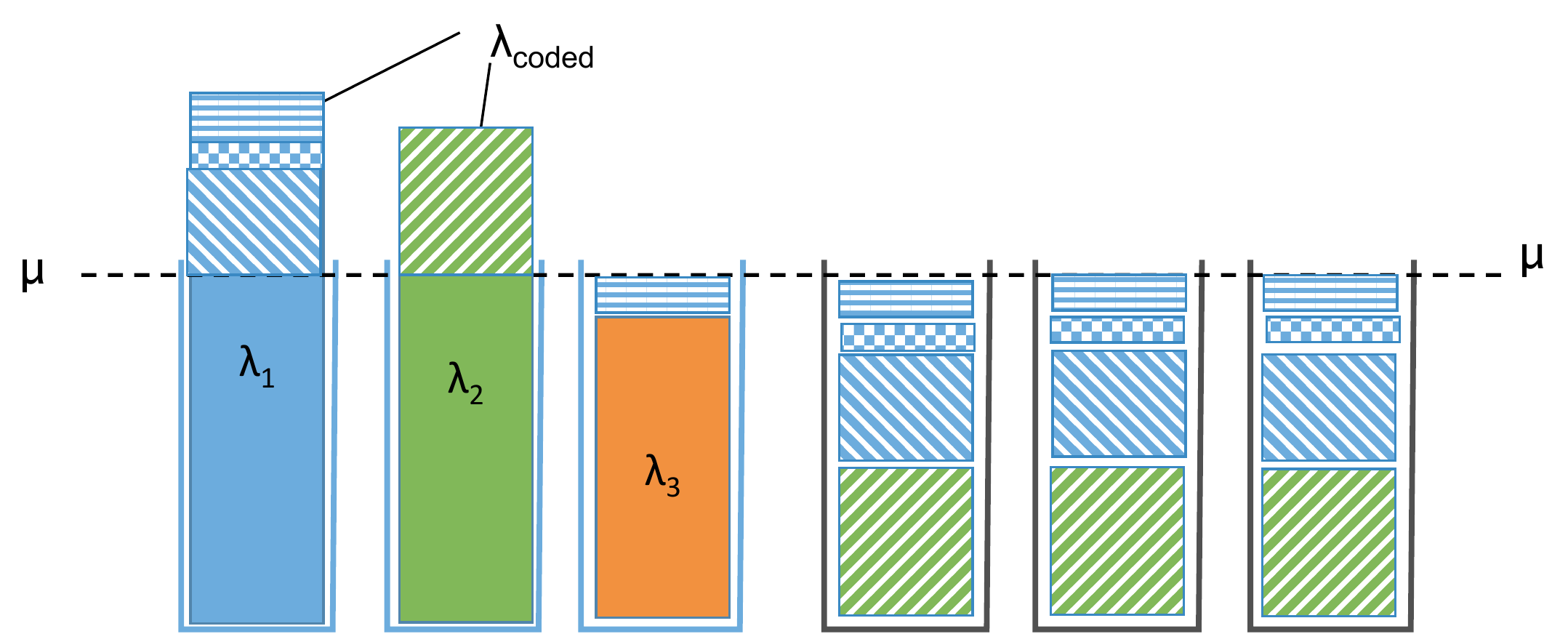}
\caption{Water-filling strategy to server the requests using coded nodes\label{fig:mds_waterfill_2}}
\vspace{-0.3cm}
\end{figure}

\subsection{$N-K \geq K$: Achievable Region Matches Outer Bound}
We seek to find a strategy to split the download requests across the $N$ servers such that the set of feasible arrival rates matches, or comes close to the outer bound given by \Cref{thm:mds_outer_bnd}. We now propose a water-filling algorithm to schedule requests to servers on a $(N, K)$ coded system. Then we prove the optimality of this algorithm by considering two cases: 1) $N-K \geq K$ (the code rate $\leq 1/2$), and 2) $N-K < K$ (the code rate  $> 1/2$).

\begin{defn}[Waterfilling Algorithm]
Given arrival rates $\lambda_1$, $\lambda_2$, \dots $\lambda_K$ for the $K$ files, the water-filling algorithm assigns them to the $N$ nodes as follows.
\begin{itemize}
\item Let $\gamma_{i}$ be the load on node $i$ for $i = 1, 2, \dots N$. Assign requests to their respective systematic nodes until these nodes are saturated. Set $\gamma_{i} = \min(\lambda_i, \mu)$ for $i =1, \dots K$.
\item Each of the remaining $\lambda_{coded} = \sum_{i=1}^K(\lambda_i - \mu)^+$ requests can be served by any $K$ unsaturated servers.
\item While $\lambda_{coded}  > 0$ and $\min_i \gamma_i < \mu$ do the following:
\begin{itemize}
\item Find the $K$ least-loaded servers in the system, that is, the $K$ servers with minimum $\gamma_i$'s. 
\item From $\lambda_{coded}$, send an infinitesimally small rate $\epsilon>0$ to each of these $K$ servers. Decrement $\lambda_{coded}$ by $\epsilon$, and increment the corresponding $K$ $\gamma_i$'s by $\epsilon$.
\end{itemize}
\end{itemize}

\end{defn}

\begin{thm}
\label{thm:waterfilling_opt}
For $N-K \geq K$, the proposed water-filling algorithm is optimal. The set of feasible arrival rates span the whole region inside the outer bound given in \Cref{thm:mds_outer_bnd}.
\end{thm}
\begin{proof}
For $N-K \geq K$ let us evaluate the set of feasible arrival rates using this waterfilling algorithm.

Without loss of generality, sort the arrival rates in descending order such that $\lambda_1 \geq \lambda_2 \geq \dots \geq \lambda_K$. After sending requests to systematic servers until they are saturated, the total residual arrival rate is $\lambda_{coded} = \sum_{i=1}^K(\lambda_i - \mu)^+$, as illustrated in \Cref{fig:mds_waterfill_2} for the $(6,3)$ MDS coded system. Assume that $\lambda_1 \geq \mu$. If this is not true, then $\lambda_{coded} = 0$ and all requests can be served by systematic servers.

The algorithm first uniformly splits $\min(\gamma_{K}(N-K)/k, \lambda_{coded})$ requests over $N-K$ servers, $K+1 ,\dots N$.  Then, for every $r = k, \dots, 2$, it uniformly splits $\min( (\gamma_{r-1}- \gamma_{r})(N-r+1)/k, \lambda_{coded})$ requests over $N-r+1$ servers, $r, r+1, \dots N$. Using this water-filling algorithm, the maximum rate of requests that can be supported using coded servers is
\begin{align}
\lambda_{max} &= \min(\lambda_K, \mu) \frac{N-K}{K} + \\
& \quad (\min(\lambda_{K-1},\mu) - \min(\lambda_{K}, \mu))\frac{N-K+1}{K}+ \dots +\\
& \quad (\min(\lambda_{1},\mu) - \min(\lambda_{2}, \mu))\frac{N-1}{K} \\
& = \min(\lambda_1, \mu) \frac{N}{K} - \sum_{i=1}^{K} \min(\lambda_i, \mu) \frac{1}{K} \\
&= \mu \frac{N}{K} - \sum_{i=1}^{K} \min(\lambda_i, \mu) \frac{1}{K}
\end{align}
In \Cref{fig:mds_waterfill_2}, the height of each patterned fill in the rightmost column, starting from the bottom upwards, corresponds to each term in the above summation.

The residual rate $\lambda_{coded}$ should be less than the total remaining service capacity using non-systematic servers.
\begin{align}
 \sum_{i=1}^{K}(\lambda_i - \mu)^+ &\leq  \mu \frac{N}{K} - \sum_{i=1}^{K} \min(\lambda_i, \mu) \frac{1}{K}
\end{align}
Rearranging, this is equivalent to \eqref{eqn:mds_outer_bnd}. Thus, for $N-K \geq K$, the waterfilling can achieve the region given by the outer bound in \Cref{thm:mds_outer_bnd}. Hence, it is optimal for $N-K \geq K$.
\end{proof}

\subsection{$N-K < K$: Waterfilling is optimal}
Next let us consider the second case $N-K < K$. For this case, we cannot always achieve the the same rate region as given by the outer bound. However, we can show that the waterfilling algorithm is optimal, and no other rate splitting scheme can yield a strictly larger service capacity region. This result follows from the two lemmas below.

\begin{lem}
It is optimal to first send requests to their systematic node. Only when the systematic node is saturated, requests should be served using coded servers.
\end{lem}
\begin{proof}
For $N-K < K$, we show that not utilizing the systematic node can only add load to the system, and thus reduce its service capacity region. Suppose $\lambda_i < \mu$ for some $i$, that is all requests for $f_i$ can be served by the systematic node. Instead, suppose we serve $\lambda_i - \epsilon$ rate using the systematic node $i$, and send the remaining $\epsilon$ portion to $K$ other servers, and decode file $f_i$ from the coded versions. As a result we are reducing the load on the systematic node by $\epsilon$, and instead adding $\epsilon$ load to $K$ other servers. If $N-K < K$, at least one of these $K$ servers is also a systematic node, which stores file $f_j$. Thus, the maximum rate of requests for file $f_j$ that can be served by its systematic node reduces by $\epsilon$.

For $N-K > K$, we showed in \Cref{thm:waterfilling_opt} that the water-filling algorithm, which first sends requests to the systematic node is optimal. Thus, there is no loss of optimality in sending requests to the systematic node until it is saturated.

\end{proof}

\begin{lem}
After the systematic node is saturated, it is optimal to always send each request to the $K$ least-loaded servers that can serve it.
\end{lem}
\begin{proof}
For each $\epsilon > 0$ rate of requests in $\lambda_{coded}$, we pick $K$ servers that will serve it. By using any algorithm for picking the $K$ servers, we will reach one of $2$ possible states:
\begin{enumerate}
\item $R \geq K$ unsaturated servers with the same load $\gamma < \mu$. Then we can split a maximum of $(\mu - \gamma)R/K$ request rate uniformly over these servers. As a result all servers will be saturated, and the outer bound will be achieved.
\item There are exactly $K$ unsaturated servers in the system with loads $\gamma_1 \geq \gamma_2 \geq \gamma_3 \geq \dots \geq \gamma_K$, where at least one of these inequalities is strict. Then the additional rate we can serve is $\mu -\gamma_1$. This would leave a non-zero amount of capacity unused.
\end{enumerate}
Since it always sends requests to the $K$ least-loaded nodes in the system, the water-filling algorithm always achieves the first state when it is feasible. And if the system ends up in the second state, water-filling minimizes $\gamma_1$.
\end{proof}

\section{Binary Simplex coded systems}
\label{sec:simplex}
Simplex codes are important subclass of availability codes. When files $f_1, \ldots, f_K$ are encoded with a binary $(N, K)$ simplex code, $N = 2^K-1$ must hold and a particular file $f_i$ can be recovered from $2^{K-1}-1$ (availability) disjoint groups of two (locality) servers. As an example, a $(7,3)$ simplex code encodes three files $\{f_1, f_2, f_3\}$ into seven as $\{f_1, f_2, f_3, f_1+f_2, f_1+f_3, f_2+f_3,f_1+f_2+f_3\}$. This code has availability three, e.g., file $f_1$ can be can be repaired from either $f_2$ and $f_1+f_2$ or $f_3$ and $f_1+f_3$, or $f_2+f_3$ and $f_1+f_2+f_3$.

Each file can be recovered from its systematic or any of its $2^{K-1}-1$ repair groups. Therefore, the request for each file can be served at rate $2^{K-1}\mu$ when the requests for all other files are zero.
\begin{lemma}
  Maximum sum of arrival rates $\lambda_1 + \ldots + \lambda_K$ that can be served by $(N, K)$ Simplex system is $2^{K-1}\mu$.
\label{lm_simplex_converse}
\end{lemma}
\begin{proof}
  \begin{figure}[hbt]
    \centering
    \begin{tikzpicture}
      \node at (0,0) {\includegraphics[scale=1]{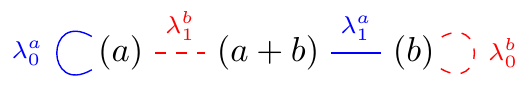}};
      \node at (0,-3) {\includegraphics[scale=1]{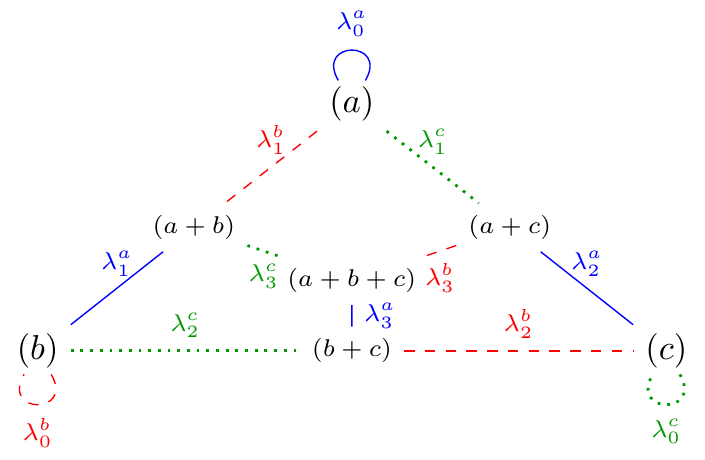}};
    \end{tikzpicture}
    \vspace*{-0.15cm}
    \caption{Graph representation of Simplex code for $K=2$ (Top) and $K=3$ (Bottom), that is inspired by the Fano plane~\cite{FanoPlane:Weisstein17}. Files are denoted as $a, b, \ldots$. Vertices correspond to servers and edges refer to repair of a file using either systematic servers (loops) or repair groups (edges between two vertices)}
  \label{fig:fig_fano_simplex}
  \end{figure}

  Fig.~\ref{fig:fig_fano_simplex} shows graph representation (i.e., Fano plane) of Simplex code for $K=2, 3$. Vertices correspond to servers and file stored on each server is indicated by its label. Each edge corresponds to service of a particular file from a systematic server (reflective loops on a server) or a repair group (edges between two servers). Recall that repairing a file from one of its repair groups requires accessing two servers, hence supplying a unit of service rate from a repair group consumes twice the capacity of supplying it from a systematic server. Service rates are shown with label $\lambda_i^f$'s on each edge such that $f \in \{a, b, \ldots\}$ denotes the file that is served and $i \in \{0, 1, 2, \ldots\}$ is an index to differentiate between the edges that serve the same file. Sum of the service rates supplied from edges that share the same vertex cannot be greater than $\mu$.
  
  Firstly consider $K=2$ system. Total service rate supplied by the system is
  \[\lambda_a + \lambda_b = \lambda_0^a + \lambda_1^a + \lambda_0^b + \lambda_1^b\]
  which is the sum of the service rates supplied by all the edges in $K=2$ graph. Edges with label $\lambda_0^a$ and $\lambda_1^b$ are attached to vertex $(a)$, hence $\lambda_0^a + \lambda_1^b \leq \mu$. Similarly, $\lambda_0^b + \lambda_1^a \leq \mu$. Thus
  \[\lambda_a + \lambda_b \leq 2\mu\]
  
  Secondly consider $K=3$ system. Total service rate supplied by the system is
  \[\lambda_a + \lambda_b + \lambda_c = \sum_{i=0}^3 \lambda_i^a + \lambda_i^b + \lambda_i^c\]
  All edges in the graph are covered by the edges attached to vertices $(a)$, $(b)$, $(c)$ and $(a+b+c)$. Therefore we can conclude
  \[\lambda_a + \lambda_b + \lambda_c \leq 4\mu\]
  
  In general, total number of edges $|E|$ in the graph of a Simplex code for $K \geq 2$ can be written as
  \[ |E| = \frac{|V|K + K}{2}\]
  where number of vertices $|V| = 2^K - 1$.
  
  Simplex is a binary linear code with generator matrix consisting of all size-$K$ bit vectors up to but not including vector of all ones. For instance,  for $K=3$, the generator matrix is
  \[
     \begin{bmatrix}
    0 & 0 & 0 & 1 & 1 & 1 & 1 \\
    0 & 1 & 1 & 0 & 0 & 1 & 1 \\
    1 & 0 & 1 & 0 & 1 & 0 & 1 \\
\end{bmatrix} \]
  Every vertex in the graph of a Simplex code can be associated with the corresponding bit vector. For instance for $K=3$, $a = [0, 0, 1]$, $b = [0, 1, 0]$ and $a+b = [0, 1, 1]$. Ignoring the loops on systematic vertices, there is an edge between two vertices if and only if corresponding bit vectors differ in a single bit (so that a symbol can be repaired from the two vertices). Then, graph of any Simplex code is a bipartite graph such that vertices that correspond to bit vectors with even number of ones can be separated from those with odd number of ones. All edges (excluding the loops) are covered by either one of the partitions. To cover also the loops, we need to pick the partition that includes the systematic vertices (i.e., bit vectors with a single one). In this chosen partition, every vertex has $K$ edges attached and no two vertices share any edge, therefore, number of vertices in the partition is $|E|/K = 2^{K-1}$.
  
  Overall, for any $K \geq 2$, there exists of a set of $2^{K-1}$ vertices that cover all the edges in the graph. Then total service rate that can be supplied by the system can be bounded as
  \[\lambda_1 + \lambda_2 + \ldots + \lambda_K \leq 2^{K-1}\mu\]
\end{proof}

Using \Cref{lm_simplex_converse}, we show that capacity region of an $(N, K)$ Simplex system is the simplex geometry in $R^K$.
\begin{theorem}
  $(N, K)$ Simplex system can serve arrival rates $\lambda_1, \ldots, \lambda_K$ if and only if $\lambda_1 + \dots + \lambda_K \leq 2^{K-1}\mu$.
\label{eq:simplex}
\end{theorem}
\begin{proof}
  If every server (systematic or not) dedicates the fraction $\alpha_i/2^{K-1}$ of its capacity solely to serving requests for file $f_i$, then the part of the system dedicated to $f_i$ acts as a $(N, K)$ binary simplex code on $2^K-1$ servers, each with capacity $\alpha_i/2^{K-1}\mu$ serving exclusively requests for file $f_i$, giving the supplied service rate of $\lambda_i = 2^{K-1}\mu\alpha_i$. By construction, inequality $\alpha_1 + \ldots + \alpha_K \leq 1$ always has to hold, and thus every achievable service rate tuple $(\lambda_1, \ldots, \lambda_K)$ can be realized by the corresponding choice of fraction tuple $(\alpha_1, \ldots, \alpha_K)$. This observation together with \Cref{lm_simplex_converse} shows that achievable capacity region of the system is a simplex in $R^K$.
\end{proof}

\section{Effect of Adding Systematic Nodes}
\label{sec:adding_sys_nodes}


Suppose we have $K=2$ files $a$ and $b$ stored across a storage system of $N$ cache nodes. Denote the arrival rates of requests for $a$ and $b$ as  $\lambda_a$ and $\lambda_b$, respectively. 
In all that follows,  we assume any $2$ coded nodes or a coded node and systematic node may recover file $a$ and file $b$. The service capacity region will be denoted by $\mathcal{S}$. Moreover,  $A$ is the number of systematic nodes for file $a$, $B$ is the number of systematic nodes for file $b$, and $C$ is the number of coded nodes.  In this section, we identify the service capacity region of such storage systems.


Let $\lambda^*_a$ denote the maximum demand for $a$ that can be supported by a given storage system.  Thus,  there exists some splitting strategy for  requests to the storage system that handles demand  $\lambda^*_a$ for file $a$.   For every $\lambda_a \leq  \lambda^*_a$ this guaranteed splitting strategy also supports demand $(\lambda_a, 0)$.  Also,  given  expected wait time $\mu$ for each of the $N$ nodes, any demand $\lambda_b>N \mu$ cannot be supported by the storage system.  In this way,  given any fixed demand $\lambda_a\leq \lambda^*_a$, the set of all of supported $\lambda_b$ is a non-empty, closed subset of $\mathbb{R}$ that is bounded above by  $N\mu$. Thus, there is a maximum such $\lambda_b$, with $(\lambda_a,\lambda_b)$ in the service capacity region of the storage system. Define $L(\lambda_a)$ to be this maximum supported $\lambda_b$  at given $\lambda_a$.  With this definition, the function
\begin{align}
L :\quad &[0,\lambda^*_a]\rightarrow\mathbb{R}\nonumber\\
& \lambda_a \mapsto L(\lambda_a)
\end{align}
is well-defined.  The storage system's service capacity region  can be described as the subset of $ \mathbb{R}^2$ that is bounded by   $\lambda_a= 0$, $\lambda_a = \lambda^*_a$, $\lambda_b = 0$ , and $\lambda_b = L(\lambda_a)$.   For convenience, further denote  $\lambda^*_b = L(0)$.

\begin{lemma}
\label{allcoded}
If  $A=B=0$ and there are $C>1$ coded nodes, then $\mathcal{S}$ \rmv{If there are no systematic nodes and $C>1$ coded nodes \rmv{such that any $2$ nodes may recover file $a$ and file $b$},
then the service capacity region} 
is the region bounded by $\lambda_a=0,\lambda_b=0,$ and  $\lambda_b=\frac{N}{2}\mu-\lambda_a$. If there are $C\leq 1$ coded nodes and no systematic nodes, then the service capacity region is the point $(0,0).$
\end{lemma}
\begin{proof}  
Since each node can support rate $\mu$ of arrivals and recovering either file requires the use of two coded nodes, 
$\lambda_a\leq \frac{C}{2}\mu$ and $\lambda_b\leq \frac{C}{2}\mu$. If $C\leq 1$, no file can be recovered. Suppose $C>1$ and label the nodes $1,\dots,C$. For each  $i=1,\dots,C$, pair node  $i$ with node $i+1\mod C$. Note that each node is in two pairs. Requests for file $a$ may be evenly divided among the $C$ pairs, and requests for file $b$ can utilize the remaining  $\mu - 2\frac{\lambda_a}{C}$ capacity of each node, with half of the node's remaining capacity devoted to each of the two pairs the node is in. Thus, the maximum achievable $\lambda_b$ is \begin{align} \lambda_b=C\frac{1}{2}\left(\mu - 2\frac{\lambda_a}{C}\right)=\frac{C}{2}\mu - \lambda_a. \end{align}
\end{proof}

Note that the conclusion of Lemma \ref{allcoded} could  be expressed as  $\lambda^*_{a}=\frac{C}{2} \mu$ and  $L(\lambda_a) = \frac{C}{2} \mu -\lambda{a}$. Also, $\lambda^*_{b}=\frac{C}{2}\mu$. In this case, the boundary $\lambda_a = \lambda^*_a$ is redundant.

\begin{lemma}
\label{adding}

%
%
%
%

\rmv{Let there be $A$ systematic nodes for file $a$, $B$ systematic nodes for file $b$, and $C$ coded nodes. Assume any $2$ coded nodes or a coded node and systematic node may recover file $a$ and file $b$. Let $\mathcal{S}$ denote the service capacity region,}
\rmv{Let $(\lambda_a^*, 0)$ denote the $\lambda_a$-intercept of $\mathcal{S}$,  $L(\lambda_a)$ denote the boundary of $\mathcal{S}$ that corresponds to the maximum achievable $\lambda_b$, and $\lambda^*_b = L(0)$.} 
Given a storage system $\mathcal{S}$ with $N = A + B + C$ nodes:

%

\begin{itemize}
\item Case 1: If $A < C$ and a systematic node is added for file $a$, then the service capacity region
$\mathcal{S'}$ has $\lambda_b$-bound

\begin{align}
L'(\lambda_a) =
\begin{cases}
-\frac{1}{2}\lambda_a +  \frac{A+C+1}{2}\mu + B\mu, & \text{$0 \leq \lambda_a \leq \mu$} \\
L(\lambda_a - \mu), & \text{$\mu \leq \lambda_a \leq \lambda_a^* + \mu.$}
\end{cases}
\end{align}

\item Case 2: If $A \geq C$ and a systematic node is added for file $a$, then the service capacity region
$\mathcal{S'}$ has $\lambda_b$-bound

\begin{align}
L'(\lambda_a) =
\begin{cases}
\lambda_b^*, & \text{$0 \leq \lambda_a \leq \mu$} \\
L(\lambda_a - \mu), & \text{$\mu \leq \lambda_a \leq \lambda_a^* + \mu.$}
\end{cases}
\end{align}

\item Case 3: If $B < C$ and a systematic node is added for file $b$, then the service capacity region
$\mathcal{S'}$ has $\lambda_b$-bound $L'(\lambda_a) =$

\begin{align}
\begin{cases}
L(\lambda_a) + \mu, & \text{$0 \leq \lambda_a \leq \lambda_a^*$} \\
-2\lambda_a + (2A + B + C + 1)\mu, & \text{$\lambda_a^* \leq \lambda_a \leq \lambda_a^* + \frac{\mu}{2}.$}
\end{cases}
\end{align}

\item Case 4: If $B \geq C$ and a systematic node is added for file $b$, then the service capacity region
$\mathcal{S'}$ has $\lambda_b$-bound  
$$L'(\lambda_a) = L(\lambda_a) + \mu\quad\text{for $0 \leq \lambda_a \leq \lambda_a^*$}.$$ 
\end{itemize}
\end{lemma}

\begin{proof}
\rmv{Let $(0, \lambda_b^*)$ denote $\lambda_b$-intercept of $\mathcal{S}$, $(\lambda_a^*, 0)$ denote $\lambda_a$-intercept of $\mathcal{S}$, and $L$ denote the boundary of $\mathcal{S}$ that corresponds to the maximum achievable $\lambda_b$.} 
\textbf{Case 1:} First, consider when $0 \leq \lambda_a \leq \mu$.

\textit{Subcase 1 ($A + C$ is odd):} Since $0 \leq \lambda_a \leq \mu$, requests for file $a$ may be divide evenly among the $A + 1$ systematic nodes for file $a$.   Since a systematic node and coded node can recover both files and $A < C$, every systematic node for file $a$ can be paired with a coded node and requests for file $b$ can utilize the remaining capacity $\mu - \frac{\lambda_a}{A+1}$. Thus, $A + 1$ coded nodes now have capacity $\mu' = \frac{\lambda_a}{A+1}$. From Lemma \ref{allcoded}, we know they can support $\frac{A + 1}{2} \mu' =  \frac{\lambda_a}{2}$ requests for file $b$. We can also pair off the remaining coded nodes, the number of which is even since $A + C $ is odd. They can be utilized their full capacity $\mu$. Note, the $B$ systematic nodes for file $b$ can support a rate $\mu$ of arrivals for file $b$. Thus, the maximum achievable $\lambda_b$ is 
\[
\lambda_b = -\frac{\lambda_a}{2}+ \frac{A + C + 1}{2}\mu + B\mu.
\]

\textit{Subcase 2 ($A + C$ is even):} Since $0 \leq \lambda_a \leq \mu$, requests for file $a$ may be divide evenly among $A$ of the systematic nodes for file $a$. As above, every systematic node for file $a$ can be paired with a coded node, and each pair can support $\frac{A}{2} \mu' =  \frac{\lambda_a}{2}$ requests for file $b$. Note, one of the systematic nodes for file $a$ has not received any requests. We can form a triple with this systematic node and two coded nodes to serve $1.5 \mu$ requests for file $b$. 
Also, the $B$ systematic nodes for file $b$ can support a rate $\mu$ of arrivals for file $b$. Thus, the maximum achievable $\lambda_b$ is the same as Subcase 1.

Now, if $\mu \leq \lambda_a \leq \lambda_a^* + \mu$, then we may send $\mu$ requests for file $a$ to the systematic node for file $a$ that was added to the system. Then the system of available nodes reduces to the previous system, $L(\lambda_a - \mu)$. 
A similar argument may be used to prove \textbf{Cases 2, 3, and 4}.
\end{proof}

Note that in Lemma \ref{adding}, the region $\mathcal{S}$ has 
$\lambda^*_b=\frac{\mu}{2}(A+C+2B)$,  so  it would be equivalent in Case 1 to specify $L'(\cdot)$ on $0\leq\lambda_a \leq \mu$ as  $L(\lambda_a) =-\frac{1}{2} \lambda_a  + \frac{1}{2}\mu+ \lambda^*_b$. In this way, when $A<C$ the addition of a systematic node for file $a$ adds a ``bonus"  region beyond the right-shift by $\mu$ that is seen both in Case 2 and when such a node is added to an uncoded system.  A similar ``bonus" region is added in Case 3.  
Figure \ref{2C1Aand2C3A} pictures the resulting rate region for systems with $B=0$. 
\begin{figure}[h!]
\begin{center}
\begin{tikzpicture}[thick,scale=.9]
\draw[->,thick] (0,0)--(2.5,0) node[right]{$\lambda_a$};
\draw[->,thick] (0,0)--(0,2) node[above]{$\lambda_b$};
\draw[fill=blue,fill opacity=0.2] (0.8,0) --(0.8,1.2)--(2.1,0)--cycle;
\draw[fill=red,fill opacity=0.2] (0,0)--(0.8,0)--(0.8,1.2)--(0,1.6)--cycle;
\node[circle,fill=black,scale=.5,label=left:$\frac{C}{2}\mu$] (x0) at (0,1.2) {};
\node[circle,fill=black,scale=.5,label=below:$A\mu$] (x1) at (0.8,0) {};
\node[circle,fill=black,scale=.5,label=below:$(\frac{C}{2}+A)\mu$] (x2) at (2.1,0) {};
\node[circle,fill=black,scale=.5,label=left:$\frac{C+A}{2}\mu$] (x3) at (0,1.6) {};
\node[circle,scale=.5,label=below left:$0$] (O) at (.1,.1) {};
\end{tikzpicture}
\begin{tikzpicture}[thick,scale=.8]
\draw[->,thick] (0,0)--(3.7,0) node[right]{$\lambda_a$};
\draw[->,thick] (0,0)--(0,2.4) node[above]{$\lambda_b$};
\draw[fill=blue,fill opacity=0.2] (2.1,0) --(2.1,1.5)--(3.3,0)--cycle;
\draw[fill=red,fill opacity=0.2] (0.8,0)--(2.1,0)--(2.1,1.5)--(0.8,2)--cycle;
\draw[fill=red,fill opacity=0.4] (0,0)--(0.8,0)--(0.8,2)--(0,2)--cycle;
\node[circle,fill=black,scale=.5,label=left:$\frac{C}{2}\mu$] (x0) at (0,1.5) {};
\node[circle,fill=black,scale=.5,label=below:$A\mu$] (x1) at (2.1,0) {};
\node[circle,fill=black,scale=.5,label=below:$(\frac{C}{2}+A)\mu$] (x2) at (3.3,0) {};
\node[circle,fill=black,scale=.5,label=left:$C\mu$] (x3) at (0,2) {};
\node[circle,fill=black,scale=.5,label=below:$(A-C)\mu$] (x3) at (0.8,0) {};
\node[circle,scale=.5,label=left:$0$] (O) at (0,0) {};
\end{tikzpicture}
\end{center}
\caption{$\mathcal{S}$ when $B=0$ and (\emph{left}) $A\leq C$,  (\emph{right}) $A>C$.}
\label{2C1Aand2C3A}
\end{figure}
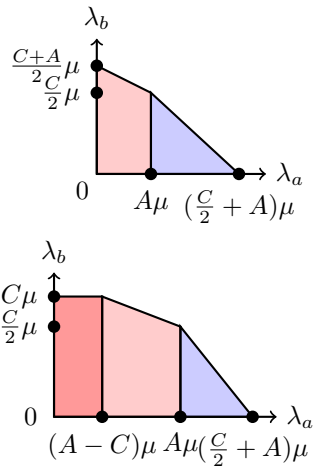

\begin{theorem}
\rmv{Let there be $A$ systematic nodes for file $a$, $B$ systematic nodes for file $b$, and $C$ coded nodes. Assume any $2$ coded nodes or a coded node and systematic node may recover file $a$ and file $b$. Then t} The service capacity region is bounded by $\lambda_a=0, \lambda_b=0, \lambda_a=\min\{(A+C)\mu, (A+\frac{B}{2}+\frac{C}{2})\mu\}$, and $L(\lambda_a) =$ \[\begin{cases} 
      (B+C)\mu & \text{if } A> C \text{ and } \\&0\leq\lambda_a\leq (A-C)\mu \\
      -\frac{1}{2}\lambda_a+(\frac{A}{2}+B+\frac{C}{2})\mu& \text{if } A >C  \text{ and }\\ & (A-C)\mu<\lambda_a\leq A\mu \\
      -\frac{1}{2}\lambda_a+(\frac{A}{2}+B+\frac{C}{2})\mu& \text{if } A\leq C \text{ and }  \\& 0\leq \lambda_a\leq A\mu\\
        -\lambda_a+(A+B+\frac{C}{2})\mu& \text{if } A\mu < \lambda_a\leq (A+\frac{C}{2})\mu  \\
      -2\lambda_a+(2A+B+C)\mu& \text{if }B> C \text{ and } \\&  (A+\frac{C}{2})\mu <\lambda_a\leq A+C \\
      -2\lambda_a+(2A+B+C)\mu& \text{if } B\leq C\text{ and } (A+\frac{C}{2})\mu \\& <\lambda_a \leq( A+\frac{B}{2}+\frac{C}{2})\mu.   
   \end{cases}
\]
\end{theorem}
\begin{proof}
Let $A'=A$ if $A\leq C$ and $A'=C$ otherwise. Let $A''=A-C$ if $A>C$ and $0$ otherwise. Similarly, let $B'=B$ if  $B\leq C$ and $B'=C$ otherwise and let $B''=B-C$ if $B>C$ and $B''=0$ otherwise.
Consider building a coded storage system by first adding $C$ coded nodes. Then by Lemma \ref{allcoded}, the service capacity region is bounded by $\lambda_a = 0$, $\lambda_b = 0$, and $L(\lambda_a) = \frac{C}{2} \mu - \lambda_a$ where $0 \leq \lambda_a \leq \frac{C}{2}\mu$. 

We will then add $A'$ systematic nodes for file $a$. Repeatedly applying Lemma \ref{adding} Case 1 yields $L(\lambda_a) =$

\begin{align}
\begin{cases}
-\frac{1}{2}\lambda_a +  \frac{A'+C}{2}\mu, & \text{$0 \leq \lambda_a \leq A'\mu$} \\
(A'+\frac{C}{2}) \mu - \lambda_a, & \text{$A'\mu \leq \lambda_a \leq  A'\mu + \frac{C}{2}\mu.$}
\end{cases}
\end{align}

If $A''\neq 0$, then we may continue to add systematic nodes for file $a$ by applying Lemma \ref{adding} Case 2. This yields $L(\lambda_a) =$
\begin{align}
\begin{cases}
 \frac{A'+C}{2}\mu & \text{if $0 \leq \lambda_a \leq A''\mu$} \\
-\frac{1}{2}\lambda_a +  \frac{A'+A''+C}{2}\mu & \text{if $A'' \mu\leq \lambda_a \leq (A'+A'')\mu$} \\
(A'+A''+\frac{C}{2}) \mu - \lambda_a & \text{if $(A'+A'')\mu \leq \lambda_a$} \\&\leq  (A'+A'')\mu + \frac{C}{2}\mu.
\end{cases}
\end{align}
Thus we have $L(\lambda_a) =$
\begin{align}
\begin{cases}
-\frac{1}{2}\lambda_a +  \frac{A+C}{2}\mu & \text{if $A\leq C$ and  $0 \leq \lambda_a \leq A\mu$} \\
 C\mu & \text{if $A>C$ and $0 \leq \lambda_a \leq (A-C)\mu$} \\
-\frac{1}{2}\lambda_a +  \frac{A+C}{2}\mu & \text{if $A>C$ and  $(A-C) \mu\leq \lambda_a \leq A\mu$} \\
(A+\frac{C}{2}) \mu - \lambda_a & \text{if $A\mu \leq \lambda_a \leq A\mu + \frac{C}{2}\mu.$}
\end{cases}
\end{align}
Applying Lemma \ref{adding} Case 3 to add $B'$ systematic nodes for file $b$ yields $L(\lambda_a) =$
\begin{align}
\begin{cases}
-\frac{1}{2}\lambda_a +  (\frac{A+C}{2}+B')\mu & \text{if $A\leq C$ and } 0 \leq \lambda_a \leq A\mu \\
 (B'+C)\mu & \text{if $A>C$ and}\\& 0 \leq \lambda_a \leq (A-C)\mu \\
-\frac{1}{2}\lambda_a + ( \frac{A+C}{2}+B')\mu & \text{if $A>C$ and}\\&  (A-C) \mu\leq \lambda_a \leq A\mu \\
(A+B'+\frac{C}{2}) \mu - \lambda_a & \text{if $A\mu \leq \lambda_a \leq A\mu + \frac{C}{2}\mu$}\\
-2\lambda_a+(2A+B'+C)\mu & \text{if $ A\mu + \frac{C}{2}\mu<\lambda_a$}\\&\leq  (A+ \frac{B'+C}{2})\mu .
\end{cases}
\end{align}
If $B''\neq 0$, we may continue to add systematic nodes for file $b$ by applying Lemma \ref{adding} Case 4. This raises the boundary of the service capacity region by $B''\mu$ in the $\lambda_b$ direction, giving the desired result.
\end{proof}
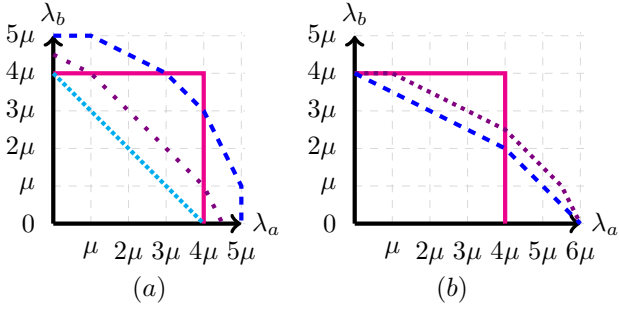
\begin{figure}
\begin{center}
\begin{tikzpicture}[ultra thick,scale=.5]
\draw[help lines, color=gray!30, dashed] (-.01,-.01) grid (5.1,5.1);
\draw[->,ultra thick,black] (0,0)--(5,0) node[right]{$\lambda_a$};
\draw[->,ultra thick,black] (0,0)--(0,5) node[above]{$\lambda_b$};
\draw[ultra thick,magenta] (0,4)--(4,4)--(4,0);
\draw[ultra thick,blue,dashed] (0,5)--(1,5)--(3,4)--(4,3)--(5,1)--(5,0);
\draw[ultra thick,violet,loosely dotted] (0,4.5)--(1,4)--(4,1)--(4.5,0);
\draw[ultra thick,cyan,densely dotted] (0,4)--(4,0);
\node[circle,scale=.5,label=left:$0$] (O) at (0,0) {};
\node[circle,scale=.5,label=below:$\mu$] (x1) at (1,0) {};
\node[circle,scale=.5,label=below:$2\mu$] (x2) at (2,0) {};
\node[circle,scale=.5,label=below:$3\mu$] (x3) at (3,0) {};
\node[circle,scale=.5,label=below:$4\mu$] (x4) at (4,0) {};
\node[circle,scale=.5,label=below:$5\mu$] (x5) at (5,0) {};
\node[circle,scale=.5,label=left:$\mu$] (y1) at (0,1) {};
\node[circle,scale=.5,label=left:$2\mu$] (y2) at (0,2) {};
\node[circle,scale=.5,label=left:$3\mu$] (y3) at (0,3) {};
\node[circle,scale=.5,label=left:$4\mu$] (y4) at (0,4) {};
\node[circle,scale=.5,label=left:$5\mu$] (y5) at (0,5) {};
\node[circle,scale=.5,label=left:$(a)$] (a) at (3.5,-1.75) {};
\end{tikzpicture}
\begin{tikzpicture}[ultra thick,scale=.5]
\draw[help lines, color=gray!30, dashed] (-.01,-.01) grid (6.1,5.1);
\draw[->,ultra thick,black] (0,0)--(6,0) node[right]{$\lambda_a$};
\draw[->,ultra thick,black] (0,0)--(0,5) node[above]{$\lambda_b$};
\draw[ultra thick,magenta] (0,4)--(4,4)--(4,0);
\draw[ultra thick,violet, dotted] (0,4)--(1,4)--(4,2.5)--(5.5,1)--(6,0);
\draw[ultra thick,blue,dashed] (0,4)--(4,2)--(6,0);
\node[circle,scale=.5,label=left:$0$] (O) at (0,0) {};
\node[circle,scale=.5,label=below:$\mu$] (x1) at (1,0) {};
\node[circle,scale=.5,label=below:$2\mu$] (x2) at (2,0) {};
\node[circle,scale=.5,label=below:$3\mu$] (x3) at (3,0) {};
\node[circle,scale=.5,label=below:$4\mu$] (x4) at (4,0) {};
\node[circle,scale=.5,label=below:$5\mu$] (x5) at (5,0) {};
\node[circle,scale=.5,label=below:$6\mu$] (x6) at (6,0) {};
\node[circle,scale=.5,label=left:$\mu$] (y1) at (0,1) {};
\node[circle,scale=.5,label=left:$2\mu$] (y2) at (0,2) {};
\node[circle,scale=.5,label=left:$3\mu$] (y3) at (0,3) {};
\node[circle,scale=.5,label=left:$4\mu$] (y4) at (0,4) {};
\node[circle,scale=.5,label=left:$5\mu$] (y5) at (0,5) {};
\node[circle,scale=.5,label=left:$(b)$] (b) at (3.5,-1.75) {};
\end{tikzpicture}
\end{center}
\caption{Boundaries of rate regions with $N=8$. 
\\(a)  $A=4,B=4,C=0$ (solid), $A=3, B=3, C=2$ (dashed), $A=1, B=1, C=6$ (loosely dotted), $A=0, B=0, C=8$ (densely dotted).
\\(b) $A=4, B=4, C=0$ (solid), $A=4, B=1, C=3$ (dotted), $A=4, B=0, C=4$ (dashed).}
\label{N8}
\end{figure}

\section{Concluding Remarks}
\label{sec:conclu}
Popular content files are generally replicated at multiple nodes in order to support a larger volume of access requests. For large files, it can be slow and expensive to dynamically add replicas to adjust to changes in popularity.

In this paper, we consider a erasure coded system where some nodes store coded combinations of multiple content files. We determine the service capacity region, or the maximum rate of access requests that can be served by this system. Our results indicate that for the same amount of redundancy, adding coded nodes instead of replicas provides more robustness to changes in content popularity. We determine the capacity region for commonly used codes like MDS and Simplex codes. Comparison of the regions sheds light on designing the erasure code to maximize service capacity.

To the best of our knowledge, this is the first work to analyze the service capacity of coded storage systems. There are many questions open for future research. While we have studied the service capacity regions of commonly used codes, developing a general theory to optimally split requests, and design codes that maximize the capacity region remains an open problem. Also, in this paper we consider sending a request to one of the repair groups. Redundantly assigning requests to multiple groups and waiting for any one copy may increase service capacity, as shown in \cite{joshi2017boosting} for task replication in computing.

\section*{Acknowledgments}
Part of this research is based upon work supported by the National Science Foundation under Grants No.~CIF-1717314, and No.~DMS-1439786 while some authors were in residence at the Institute for Computational and Experimental Research in Mathematics in Providence, RI, during the {\it Women in Data Science and Mathematics Research Collaboration Workshop} at ICERM in July 2017.




\bibliographystyle{IEEEtran}
\bibliography{swanand_bib} 

\begin{thebibliography}{10}
\providecommand{\url}[1]{#1}
\csname url@samestyle\endcsname
\providecommand{\newblock}{\relax}
\providecommand{\bibinfo}[2]{#2}
\providecommand{\BIBentrySTDinterwordspacing}{\spaceskip=0pt\relax}
\providecommand{\BIBentryALTinterwordstretchfactor}{4}
\providecommand{\BIBentryALTinterwordspacing}{\spaceskip=\fontdimen2\font plus
\BIBentryALTinterwordstretchfactor\fontdimen3\font minus
  \fontdimen4\font\relax}
\providecommand{\BIBforeignlanguage}[2]{{%
\expandafter\ifx\csname l@#1\endcsname\relax
\typeout{** WARNING: IEEEtran.bst: No hyphenation pattern has been}%
\typeout{** loaded for the language `#1'. Using the pattern for}%
\typeout{** the default language instead.}%
\else
\language=\csname l@#1\endcsname
\fi
#2}}
\providecommand{\BIBdecl}{\relax}
\BIBdecl

\bibitem{Dimakis:10}
A.~G. Dimakis, P.~B. Godfrey, M.~Wainwright, and K.~Ramachandran, ``Network
  {C}oding for {D}istributed {S}torage {S}ystems,'' vol.~56, no.~9, pp.
  4539--4551, Sep. 2010.

\bibitem{Dimakis-Survey:11}
A.~G. Dimakis, K.~Ramchandran, Y.~Wu, and C.~Suh, ``A {S}urvey on {N}etwork
  {C}odes for {D}istributed {S}torage,'' \emph{Proceedings of the {IEEE}},
  vol.~99, no.~3, pp. 476--489, Mar. 2011.

\bibitem{Huang:07}
C.~Huang, M.~Chen, and J.~Li, ``Pyramid codes: Flexible schemes to trade space
  for access efficiency in reliable data storage systems,'' in \emph{Network
  Computing and Applications, 2007. NCA 2007. Sixth IEEE International
  Symposium on}, July 2007, pp. 79--86.

\bibitem{gopalan2012locality}
P.~Gopalan, C.~Huang, H.~Simitci, and S.~Yekhanin, ``On the locality of
  codeword symbols,'' \emph{Information Theory, IEEE Transactions on}, vol.~58,
  no.~11, pp. 6925--6934, Nov 2012.

\bibitem{joshi2012coding}
G.~Joshi, Y.~Liu, and E.~Soljanin, ``Coding for fast content download,'' in
  \emph{Communication, Control, and Computing (Allerton), 2012 50th Annual
  Allerton Conference on}.\hskip 1em plus 0.5em minus 0.4em\relax IEEE, 2012,
  pp. 326--333.

\bibitem{shah2014mds}
N.~B. Shah, K.~Lee, and K.~Ramchandran, ``The {MDS} queue: Analysing the
  latency performance of erasure codes,'' in \emph{2014 IEEE International
  Symposium on Information Theory (ISIT'14)}, pp. 861--865.

\bibitem{liang2014fast}
G.~Liang and U.~C. Kozat, ``Fast cloud: Pushing the envelope on delay
  performance of cloud storage with coding,'' \emph{Networking, IEEE/ACM
  Transactions on}, vol.~22, no.~6, pp. 2012--2025, 2014.

\bibitem{gardner2015reducing}
K.~Gardner, S.~Zbarsky, S.~Doroudi, M.~Harchol-Balter, and E.~Hyytia,
  ``Reducing latency via redundant requests: Exact analysis,'' \emph{ACM
  SIGMETRICS Performance Evaluation Review}, vol.~43, no.~1, pp. 347--360,
  2015.

\bibitem{swanand_isit_2015}
S.~Kadhe, E.~Soljanin, and A.~Sprintson, ``Analyzing download time for
  availability codes,'' in \emph{Information Theory Proceedings (ISIT), 2015
  IEEE International Symposium on}, July 2015.

\bibitem{swanand_allerton_2015}
------, ``When do the availability codes make the stored data more available?''
  in \emph{2015 53rd Annual Allerton Conference on Communication, Control, and
  Computing (Allerton)}, Sept 2015, pp. 956--963.

\bibitem{Simplex:AktasNS17}
M.~F. Aktas, E.~Najm, and E.~Soljanin, ``Simplex queues for hot-data
  download,'' in \emph{Proceedings of the 2017 ACM SIGMETRICS/International
  Conference on Measurement and Modeling of Computer Systems}.\hskip 1em plus
  0.5em minus 0.4em\relax ACM, 2017, pp. 35--36.

\bibitem{Noori:16}
M.~Noori, E.~Soljanin, and M.~Ardakani, ``On storage allocation for maximum
  service rate in distributed storage systems,'' in \emph{2016 IEEE
  International Symposium on Information Theory (ISIT)}, July 2016, pp.
  240--244.

\bibitem{Wang:14}
A.~Wang and Z.~Zhang, ``Repair locality with multiple erasure tolerance,''
  \emph{Information Theory, IEEE Transactions on}, vol.~60, no.~11, pp.
  6979--6987, Nov 2014.

\bibitem{Rawat:14Availability}
A.~Rawat, D.~Papailiopoulos, A.~Dimakis, and S.~Vishwanath, ``Locality and
  availability in distributed storage,'' in \emph{Information Theory (ISIT),
  2014 IEEE International Symposium on}, June 2014, pp. 681--685.

\bibitem{Tamo:14Availability}
I.~Tamo and A.~Barg, ``Bounds on locally recoverable codes with multiple
  recovering sets,'' in \emph{Information Theory (ISIT), 2014 IEEE
  International Symposium on}, June 2014, pp. 691--695.

\bibitem{CadambeM:15}
V.~Cadambe and A.~Mazumdar, ``Bounds on the size of locally recoverable
  codes,'' \emph{IEEE Transactions on Information Theory}, vol.~61, no.~11, pp.
  5787--5794, Nov 2015.

\bibitem{FanoPlane:Weisstein17}
\BIBentryALTinterwordspacing
E.~W. Weisstein, ``Fano plane. {From MathWorld---A Wolfram Web Resource}.''
  [Online]. Available: \url{http://mathworld.wolfram.com/FanoPlane.html}
\BIBentrySTDinterwordspacing

\bibitem{joshi2017boosting}
G.~Joshi, ``Boosting service capacity via adaptive replication,'' in
  \emph{Proceedings of ACM/IFIP Performance}, Nov. 2017.

\end{thebibliography}

\vfill{}
\end{document}